\documentclass[11pt,a4paper]{article}

\usepackage[utf8]{inputenc}

\usepackage{fullpage} 

\usepackage{graphicx}

\usepackage{tabularx}

\usepackage{amsthm}
\usepackage{amsmath}
\usepackage{amssymb}
\usepackage{enumerate}
\usepackage{amsfonts}
\usepackage{mathtools}
\usepackage[sort,numbers,sectionbib]{natbib}

\usepackage{makecell}

\usepackage{tikz}
\usetikzlibrary{arrows,decorations.pathmorphing,decorations.pathreplacing,backgrounds,positioning,fit,matrix,arrows.meta,automata}
\usetikzlibrary{shapes,calc}

\tikzset{%
  apply style/.code={%
    \tikzset{#1}%
  }
}

\usepackage{dsfont} 

\usepackage{algorithm}
\usepackage[noend]{algpseudocode}
\algnewcommand\algorithmicinput{\textbf{Input:}}
\algnewcommand\algorithmicoutput{\textbf{Output:}}
\algnewcommand\Input{\item[\algorithmicinput]}
\algnewcommand\Output{\item[\algorithmicoutput]}

\usepackage{etoolbox}
\usepackage{tikz}
\usetikzlibrary{tikzmark}
\usetikzlibrary{calc}

\usepackage[menucolor=orange!40!black,filecolor=magenta!40!black,urlcolor=blue!40!black,linkcolor=red!40!black,citecolor=green!40!black,colorlinks,pagebackref]{hyperref}
\usepackage[capitalize,nameinlink,noabbrev]{cleveref}

\errorcontextlines\maxdimen

\newcommand{\ALGtikzmarkcolor}{black}
\newcommand{\ALGtikzmarkextraindent}{4pt}
\newcommand{\ALGtikzmarkverticaloffsetstart}{-.5ex}
\newcommand{\ALGtikzmarkverticaloffsetend}{-.5ex}
\makeatletter
\newcounter{ALG@tikzmark@tempcnta}

\newcommand\ALG@tikzmark@start{%
	\global\let\ALG@tikzmark@last\ALG@tikzmark@starttext%
	\expandafter\edef\csname ALG@tikzmark@\theALG@nested\endcsname{\theALG@tikzmark@tempcnta}%
	\tikzmark{ALG@tikzmark@start@\csname ALG@tikzmark@\theALG@nested\endcsname}%
	\addtocounter{ALG@tikzmark@tempcnta}{1}%
}

\def\ALG@tikzmark@starttext{start}
\newcommand\ALG@tikzmark@end{%
	\ifx\ALG@tikzmark@last\ALG@tikzmark@starttext
	\else
	\tikzmark{ALG@tikzmark@end@\csname ALG@tikzmark@\theALG@nested\endcsname}%
	\tikz[overlay,remember picture] \draw[\ALGtikzmarkcolor] let \p{S}=($(pic cs:ALG@tikzmark@start@\csname ALG@tikzmark@\theALG@nested\endcsname)+(\ALGtikzmarkextraindent,\ALGtikzmarkverticaloffsetstart)$), \p{E}=($(pic cs:ALG@tikzmark@end@\csname ALG@tikzmark@\theALG@nested\endcsname)+(\ALGtikzmarkextraindent,\ALGtikzmarkverticaloffsetend)$) in (\x{S},\y{S})--(\x{S},\y{E});%
	\fi
	\gdef\ALG@tikzmark@last{end}%
}

\apptocmd{\ALG@beginblock}{\ALG@tikzmark@start}{}{\errmessage{failed to patch}}
\pretocmd{\ALG@endblock}{\ALG@tikzmark@end}{}{\errmessage{failed to patch}}
\makeatother

\newcommand{\appsymb}{$\star$}
\newcommand{\appref}[1]{{\hyperref[proof:#1]{\appsymb}}}
\newcommand{\appLink}[1]{{\hyperref[#1]{\appsymb}}}

\newtheorem{theorem}{Theorem}[section]
\newtheorem{lemma}[theorem]{Lemma}
\newtheorem{observation}[theorem]{Observation}
\newtheorem{corollary}[theorem]{Corollary}
\newtheorem{proposition}[theorem]{Proposition}

\theoremstyle{definition}
\newtheorem{definition}[theorem]{Definition}

\crefname{observation}{Observation}{Observations}
\Crefname{observation}{Observation}{Observations}

\newcommand{\problemn}[1]{\textsc{#1}}

\DeclarePairedDelimiterX{\abs}[1]{\lvert}{\rvert}{#1}

\newcommand{\RR}{\mathbb{R}}

\newcommand{\TE}{\mathcal{E}}
\newcommand{\TG}{\mathcal{G}}

\newcommand{\bigO}{{O}}

\newcommand*{\kddpaperbetweennesshat}{Lemma~2.8}

\newcommand*{\iversonian}[1]{\left[#1\right]_{\mathds{1}}}
\newcommand*{\natinterval}[1]{\left[#1\right]}

\newcommand*{\SP}{\#P}
\newcommand*{\SPC}{\SP-hard}
\newcommand*{\SPCness}{\SP-hardness}

\newcommand*{\sourcevertex}{s}
\newcommand*{\midvertex}{v}
\newcommand*{\targetvertex}{z}

\newcommand*{\subpath}[3]{#1\left[#2, #3\right]}
\newcommand*{\sbullet}[1][.5]{\mathbin{\vcenter{\hbox{\scalebox{#1}{$\bullet$}}}}}
\newcommand*{\subpathblank}{\sbullet[.75]}
\newcommand*{\pathsuff}[2]{\subpath{#1}{#2}{\subpathblank}}
\newcommand*{\pathpref}[2]{\subpath{#1}{\subpathblank}{#2}}
\newcommand*{\pathconcat}{\oplus}

\newcommand*{\spath}[2]{temporal #1\text{-}#2\text{-path}}
\newcommand*{\spaths}[2]{\spath{#1}{#2}s}
\newcommand*{\szpath}{\spath{$s$}{$z$}}
\newcommand*{\szpaths}{\spaths{$s$}{$z$}}
\newcommand*{\swalk}[2]{temporal #1\text{-}#2\text{-walk}}
\newcommand*{\swalks}[2]{\swalk{#1}{#2}s}
\newcommand*{\szwalk}{\swalk{$s$}{$z$}}
\newcommand*{\szwalks}{\swalks{$s$}{$z$}}

\newcommand*{\timespan}{T}
\newcommand*{\temptuple}{(V, \TE, \timespan)}
\newcommand*{\appset}{V\times\natinterval{T}}
\newcommand*{\vapp}[2]{(#1, #2)}
\newcommand*{\vappdef}{\vapp{\midvertex}{t}}
\newcommand*{\vappsource}{\vapp{\sourcevertex}{0}}
\newcommand*{\trans}[1]{\overset{#1}{\rightarrow}}

\newcommand*{\pathset}{\walkset}

\newcommand*{\walkset}{\mathcal{W}}
\newcommand*{\walksetrest}[1]{\walkset_{#1}}
\newcommand*{\walksetrestdef}{\walksetrest{\sourcevertex}}

\newcommand*{\Prefcompat}{Prefix-compatibility}
\newcommand*{\prefcompat}{prefix-compatibility}

\newcommand*{\prefcomple}{prefix-compatible}

\newcommand*{\Prefopt}{Prefix-optimality}
\newcommand*{\prefopt}{prefix-optimality}

\newcommand*{\prefoptal}{prefix-optimal}

\newcommand*{\Prefext}{Prefix-exchangeability}
\newcommand*{\prefext}{prefix-exchangeability}
\newcommand*{\prefextable}{prefix-exchangeable}

\newcommand*{\btw}{C_B}
\newcommand*{\btwset}[1]{\btw^{#1}}
\newcommand*{\btwsetdef}{\btwset{\walkset}}
\newcommand*{\totbtw}{\hat{C}_B}
\newcommand*{\totbtwset}[1]{\totbtw^{#1}}
\newcommand*{\totbtwsetdef}{\totbtwset{\walkset}}

\newcommand*{\totbtwtext}{total temporal betweenness centrality}
\newcommand*{\Pairdeptext}{Pair dependency}
\newcommand*{\pairdeptext}{pair dependency}

\newcommand*{\Temppairdeptext}{Temporal \pairdeptext{}}
\newcommand*{\temppairdeptext}{temporal \pairdeptext{}}

\newcommand*{\cumdep}{cumulative dependency}
\newcommand*{\cumdeps}{cumulative dependencies}

\newcommand*{\cumtdep}{temporal \cumdep{}}
\newcommand*{\Appdeptext}{Appearance dependency}
\newcommand*{\appdeptext}{appearance dependency}

\newcommand*{\Edgedeptext}{Edge dependency}
\newcommand*{\edgedeptext}{edge dependency}

\newcommand*{\walks}[1]{\text{walks}(#1)}

\newcommand*{\sigmares}[1]{\sigma^{#1}}
\newcommand*{\sigmaresdef}{\sigmares{\walkset}}

\newcommand*{\dirpredset}[2]{\operatorname{Pre}^{#1}_{#2}}
\newcommand*{\dirpredsetdef}{\dirpredset{\pathset}{\sourcevertex}}
\newcommand*{\dirsuccset}[2]{\operatorname{Succ}^{#1}_{#2}}
\newcommand*{\dirsuccsetdef}{\dirsuccset{\pathset}{\sourcevertex}}
\newcommand*{\predgraph}[1]{G^{\dirpredset{}{}}_{#1}}
\newcommand*{\predgraphdef}{\predgraph{\sourcevertex}}

\newcommand*{\pairdep}[3]{\delta^{#1}_{#2#3}}
\newcommand*{\pairdepdef}{\pairdep{\pathset}{\sourcevertex}{\targetvertex}}

\newcommand*{\deltapoint}[2]{\pairdep{#1}{#2}{\bullet}}
\newcommand*{\deltapointdef}{\deltapoint{\pathset}{\sourcevertex}}

\newcommand*{\appdepdef}{\pairdep{\pathset}{\sourcevertex}{\midvertex}(\midvertex, t)}

\newcommand*{\countwalksfn}{\textsc{Count-Walks}}
\newcommand*{\relverarr}{R}

\newcommand*{\critfn}{c}
\newcommand*{\critfnopt}{\critfn^{*}_{\sourcevertex}}
\newcommand*{\Critfntext}{Cost function}

\newcommand*{\critfntext}{cost function}
\newcommand*{\critfnstext}{\critfntext{}s}
\newcommand*{\copt}{$\critfn$-optimal}

\newcommand*{\pathenc}[1]{\langle #1\rangle}
\newcommand*{\dist}{\textsc{cur-best}}

\usepackage{authblk}

\title{Towards Classifying the Polynomial-Time Solvability of Temporal Betweenness Centrality\thanks{An extended abstract of this work appears in the proceedings of the \emph{47th International Workshop on Graph-Theoretic Concepts in Computer Science (WG~2021)}~\cite{RMNN21}.}}

\author[1]{Maciej Rymar\thanks{Partially supported by the DFG,
project MATE (NI 369/17).}}
\author[2]{Hendrik Molter\thanks{Supported by the DFG,
project MATE (NI 369/17), and by the ISF, grant No.~1070/20. Main part of this work was done while the author was affiliated with TU~Berlin.}}
\author[1]{Andr\'e Nichterlein}
\author[1]{Rolf Niedermeier}

\date{ }

\affil[1]{TU Berlin, Algorithmics and Computational Complexity, Berlin, Germany,  
\texttt{\{m.rymar,andre.nichterlein,rolf.niedermeier\}@tu-berlin.de}}

\affil[2]{Department of Industrial Engineering and Management, Ben-Gurion~University~of~the~Negev, 
Beer-Sheva, 
Israel, 
\texttt{molterh@post.bgu.ac.il}}

\begin{document}

\maketitle

\begin{abstract}
In static graphs, the betweenness centrality of a graph vertex measures 
how many times this vertex is part of a shortest path between 
any two graph vertices. Betweenness centrality is efficiently computable and 
it is a fundamental tool in network science. Continuing and extending previous
work, we study the efficient computability of betweenness centrality in 
\emph{temporal} graphs (graphs with fixed vertex set but time-varying 
arc sets). Unlike in the static case, there 
are numerous natural notions of being a ``shortest'' temporal path (walk). 
Depending on which notion is used, it was already observed 
that the problem is~\SPC{} in some cases while polynomial-time solvable 
in others. In this conceptual work, we contribute towards classifying 
what a ``shortest path (walk) concept'' has to fulfill in
order to gain polynomial-time computability of temporal betweenness centrality.

\bigskip

\end{abstract}

\section{Introduction}
Network science is a central pillar of data science. It relies 
on spotting and analyzing important network (graph) properties. 
Betweenness centrality, introduced by  \citet{freeman_set_1977} 
and made a practical tool of high relevance by \citet{Bra01}, is a key instrument in this area, in particular in the context of social network analysis.

Informally, the \emph{betweenness centrality} of a graph vertex~$v$ correlates to the probability that~$v$ is visited by a randomly chosen shortest path. 
With the advent of investigating
dynamically changing network structures and, thus,
the growing interest in temporal graphs, studying
concepts of 
temporal betweenness centrality and their algorithmic complexity became very popular 
over the recent years~\cite{afrasiabi17,alsayed2015betweenness,BMNR20,gunturi2017scalable,habiba2007betweenness,kim_temporal_2012,SML21,tang10,tsalouchidou_2020,williams16}.

The temporal graphs we are considering have fixed vertex set and edge set(s)
that change over discrete time steps. A temporal path in such a graph has to respect time, that is, the path has to traverse edges at non-decreasing time steps. 
The study of \emph{temporal} betweenness centrality is significantly richer than in static graphs
since in temporal graphs the term ``shortest path'' may have numerous 
different but natural interpretations. Indeed, the 
shortest transfer from a start vertex to a target vertex
may even be a walk (and not just a path) and there is 
also an intensive study on shortest-path and shortest-walk computations
in temporal graphs~\cite{wu_efficient_2016, himmel_efficient_2020, HMZ20,
xuan_computing_2003}. We refrain 
from going into the details here but refer to our predecessor 
work~\cite{BMNR20} for a more extensive discussion. What is important 
to note, however, is that the complexity of 
temporal betweenness centrality computation, which essentially boils down 
to a counting problem, crucially depends on the concept used. More specifically,
the complexity may vary from polynomial-time solvable (with different 
polynomial degrees) to \SPC{}. To systematically investigate this 
issue and to develop a better understanding of when one has 
to expect such a huge jump in computational complexity 
is the main motivation of our work.

The by far closest reference point for our work is a 
previous paper from our group~\cite{BMNR20}. It also 
surveys the literature roughly till the year~2020. Since then,
\citet{SML21} studied a continuous-time scenario and betweenness 
based on shortest paths, while we focus on discrete time.
Our work has a significantly stronger conceptual objective 
than \citet{BMNR20} had. So our classification results comprise
the results there. They are based on
coining the concept of \emph{prefix-compatible} temporal walks. These walks can 
be counted in polynomial time and thus the corresponding 
temporal betweenness centrality value can be computed in polynomial time.
To this end, we provide simple (still tunable) polynomial-time algorithms 
that apply to a whole class of temporal betweenness centrality 
problems. Moreover, we indicate that slightly relaxing from prefix-compatibility
typically already yields \SPCness{}.

\section{Preliminaries}\label{sec:prelims}

The fundamental concept we use in this work is a \emph{temporal
graphs}.
A directed \emph{temporal graph}~$\TG$ is a triple $\temptuple$ such that 
$V$ is a set of vertices,
$\TE \subseteq \{(u,v,t) \mid u,v \in V,u\neq v,t \in [\timespan]\}$ is a set of
time arcs, and $T\in \mathbb{N}$, where $[\timespan] \coloneqq \{1,\dots,\timespan\}$ is a set of time steps; see \cref{fig:def-illustration-paths} for an illustration.

Throughout this work, let $n:=|V|$ and $M:=|\TE|$. 
We call $V \times [\timespan]$ the set of (possible) \emph{vertex appearances}.
We consider directed temporal graphs as temporal paths and walks are implicitly directed because of the ascending time labels. 
We call a time arc $e = (v,w,t)$ also the \emph{transition} from $v$ to $w$ at time step $t$. 
We call $v$ the starting point and $w$ the endpoint of the transition. 
Using this, we can now define temporal walks and temporal paths; see \cref{fig:def-illustration-paths} for an illustration.

\tikzstyle{edge-label}=[align=center,inner sep=1pt,rounded corners,text=black!70]
\tikzstyle{pathOne}=[line join=round, draw=green!50!black,line width=5pt,line cap=round]
\tikzstyle{pathTwo}=[line join=round, draw=red,line width=5pt,line cap=round]
\tikzstyle{pathThree}=[line join=round, draw=blue,line width=5pt,line cap=round]
\newcommand{\fullGraph}[6]{
	\foreach[count=\i] \x / \y / \opt in {0/0/{label=below:{$s$}}, 1/0/{#1}, 2/0/{#2}, 3/0/{label=below:{$v$}}, 1.5/1/{#3}, 1.5/-1/{#4}, 4/1/{#5}, 4/-1/{#6},5/0/{label=below:{$z$}}}
	{
		\node (v\i) at (\x,\y) [draw=black,circle,thick,inner sep=1pt,fill=black,apply style/.expand once=\opt] {};
	}
	\begin{pgfonlayer}{background}
		\foreach \i / \j / \txt in {1/2/{1}, 1/5/{1}, 1/6/{2}, 2/3/{3}, 3/4/{5}, 4/8/{6,9}, 5/7/{4}, 6/4/{5}, 8/7/{7}, 8/9/{10}, 7/4/{8}, 7/9/{11}}
		{
			\draw [->,>=stealth,color=black!75] (v\i) edge node[fill=white,edge-label] {\small\txt}  (v\j);
		}
	\end{pgfonlayer}
}
\newcommand{\standardGraph}{\fullGraph{}{}{}{}{}{}}
\newcommand{\labelledGraph}{\fullGraph{label=below:{$a$}}{label=below:{$b$}}{label=above:{$c$}}{label=below:{$d$}}{label=above:{$e$}}{label=below:{$f$}}}

\begin{figure}[t]\centering
	\begin{tikzpicture}[xscale=1.25,yscale=1.2]
		\standardGraph
		\begin{pgfonlayer}{background}
			\begin{scope}[opacity=.15, transparency group]
				\draw [pathOne] (v1.center) \foreach \i in {5,7,9} {-- (v\i.center) };  
			\end{scope}
			\begin{scope}[opacity=.15, transparency group]
				\draw [pathThree] (v1.center) \foreach \i in {6,4,8,7,4,8,9} {-- (v\i.center) };  
			\end{scope}
		\end{pgfonlayer}
	\end{tikzpicture}
	\hfil
	\begin{tikzpicture}[xscale=1.25,yscale=1.2]
		\standardGraph
			\begin{scope}[opacity=.15, transparency group]
				\draw [pathOne] (v1.center) \foreach \i in {5,7,4,8,9} {-- (v\i.center) };  
			\end{scope}
			\begin{scope}[opacity=.15, transparency group]
				\draw [pathTwo] (v1.center) \foreach \i in {2,3,4,8,9} {-- (v\i.center) };  
			\end{scope}
			\begin{scope}[opacity=.15, transparency group]
				\draw [pathThree] (v1.center) \foreach \i in {6,4,8,9} {-- (v\i.center) };  
			\end{scope}
	\end{tikzpicture}

	\caption{
		Our running example for a temporal graph~$\TG$ with 9 vertices and 13 time-arcs (the outgoing arc from vertex~$v$ denotes two time-arcs at time steps 6 and 9).
		The number(s) on the arcs denote the time steps.
		\emph{Left:} Highlighted are the unique shortest $s$-$z$-path (top path in green) in~$\TG$ and a fastest $s$-$z$-walk (bottom walk in blue, $v$ and its successor appear twice in the walk).
		\emph{Right:} Three foremost $s$-$z$-paths in~$\TG$ are highlighted ($\TG$ has two more foremost $s$-$z$-walks; for visibility not highlighted).
	}
	\label{fig:def-illustration-paths}
\end{figure}
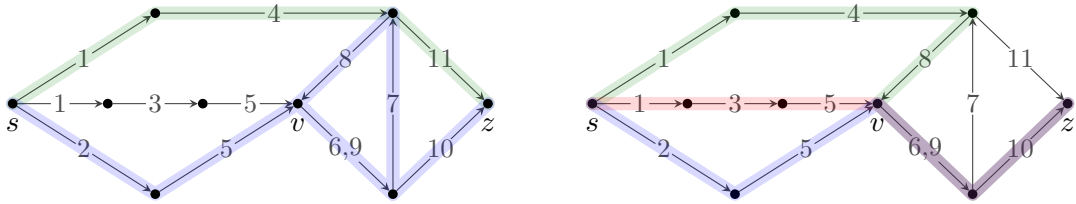

\begin{definition}[Temporal Walk]
A \emph{temporal walk} $W$ 
is an ordered sequence of transitions $(e_1,\dots,e_k)\in \TE^k$ such that for each $i\in[k-1]$
the endpoint of $e_i$ is the starting point of $e_{i+1}$ and $t_i \leq
t_{i+1}$, where $t_i$ and $t_{i+1}$ are the time labels of transitions $e_i$ and $e_{i+1}$, respectively. 
The length of~$W$ is~$\operatorname{length}(W):=k$.
\end{definition}
Let $W=(e_1,\dots,e_k)$ be a temporal walk. 
We call $W$ a \emph{\swalk{$s$}{$z$}} if $e_1=(s,v,t)$ and $e_k=(w,z,t')$ for some $v,w$ and some $t,t'$ and we call $W$ a \emph{\swalk{$s$}{$(z,t')$}} if~$e_1=(s,v,t)$ and $e_k=(w,z,t')$ for some $v,w$ and some $t$. 

A temporal walk may visit the same vertex more than once. 
In contrast to that, a temporal \emph{path} visits each vertex at most once. 
This is analogous to static graphs.
In contrast to the static setting, there are several canonical notions of ``optimal'' temporal walks. 
The three most important ones~\cite{xuan_computing_2003} are 
\begin{itemize}
\item \emph{shortest}
temporal walks, which are temporal walks between two vertices that uses a
minimum number of time arcs, 
\item \emph{fastest} temporal walks, which are temporal
walks between two vertices with a minimum difference between the time steps of
the first and last transition used by the walk, and 
\item \emph{foremost} temporal
walks, which are temporal walks between two vertices with a minimum time step
on their last transition. 
\end{itemize}
See \cref{fig:def-illustration-paths} for an illustration.
The corresponding optimal temporal paths are defined analogously. 
We remark that in the case of ``shortest'', every shortest temporal walk is in fact a temporal path, similarly to the static case.\footnote{In fact, all optimal temporal path concepts (we are aware of) where path counting and computing the betweenness centrality can be done in polynomial time have this property, ensuring that optimal walks are indeed paths.} 
However, for ``fastest'' and ``foremost'' temporal walks this is generally not the case.

For readability, we use the notation $v\trans{t}w$ instead of the triple $(v,w,t)$. 
If a temporal walk~$W$ contains the transition $v\trans{t}w$, then we say that $W$ \emph{visits} (or \emph{goes through}) vertex appearance $(w,t)$ (in \cref{fig:def-illustration-paths} (left) the blue walk visits~$(v,5)$ and~$(v,8)$). 
Let $P = v_0 \trans{t_1} v_1 \trans{t_2} \ldots \trans{t_{\ell-1}} v_{\ell-1}
\trans{t_\ell} v_\ell$ be a temporal path. For any $0 \leq i < j \leq \ell$,
we write (if $i=0$, then we define the following for $t_0=0$):
\begin{align*}
	\subpath{P}{\vapp{v_i}{t_i}}{\vapp{v_j}{t_j}} &\coloneqq v_i \trans{t_{i+1}}
	v_{i+1} \trans{t_{i+2}} \ldots \trans{t_j} v_j, \\
	\pathpref{P}{\vapp{v_j}{t_j}} &\coloneqq v_0 \trans{t_1} v_1 \trans{t_2}
	\ldots \trans{t_j} v_j, \\
	\pathsuff{P}{\vapp{v_i}{t_i}} &\coloneqq v_i \trans{t_{i+1}} v_{i+1}
	\trans{t_{i+2}} \ldots \trans{t_\ell} v_\ell.
\end{align*}
We use an analogous notation for temporal walks. Note that for a non-strict temporal
walk $W$ the above notation may not be well-defined since the same vertex
appearance $\vappdef$ may appear more than once in $W$. In this case
we define $\subpath{W}{\vapp{v_i}{t_i}}{\vapp{v_j}{t_j}}$ to be the subwalk of
$W$ from the first appearance of $u\trans{t_i}v_i$ for any $u$ to the last
appearance of  $u\trans{t_j}v_j$ for any $u$. The cases of
$\pathpref{P}{\vapp{v_j}{t_j}}$ and $\pathsuff{P}{\vapp{v_i}{t_i}}$ are handled
analogously.

Furthermore, we use $\pathconcat$ to denote \emph{concatenations} of temporal
walks, that is, let $W, W'$ be two temporal walks such that $W$ ends in $v$ and $W'$
starts in $v$. Let $t$ be the time label on the last transition of $W$ and $t'$
the label on the first transition of $W'$. Then if $t'\ge t$%
, we denote with~$W\pathconcat W'$ the concatenation of~$W$ and~$W'$.

Given a temporal graph $\TG$, we denote by $\walks{\TG}$ the set of all temporal
walks in $\TG$. Subsequently, we will need to consider the successors (or
dually, the predecessors) of each vertex appearance on temporal walks
in some~$\walkset\subseteq\walks{\TG}$.
\begin{definition}[Direct predecessor set, direct successor set]
	Let~$\TG=\temptuple$ be a temporal graph and $\walkset\subseteq\walks{\TG}$ be
	a subset of its temporal walks. Fix a source vertex~$s\in V$.
	Let~$\walksetrestdef\subseteq\walkset$ be the set of temporal walks
	in~$\walkset$ that start in~$s$. Now let~$\vapp{w}{t'}\in\appset$ be any vertex appearance.
	Then~$\dirpredsetdef(w, t')$ is the set of all direct predecessors
	of~$\vapp{w}{t'}$ on temporal walks in~$\walksetrestdef$. Formally,
	\begin{align*}
		\dirpredsetdef(w, t') \coloneqq 
		&{} \left\{\vapp{v}{t}\in \appset \mid \exists W\in\walksetrestdef\colon u \trans{t}v\trans{t'}w\in W \right\}  \\
		&{} \cup \left\{\vappsource\mid \exists W\in\walksetrestdef\colon s\trans{t'}w\in W \right\}.
	\end{align*}
	The set~$\dirsuccsetdef(v, t)$ of successors of a vertex appearance~$\vapp{v}{t}$ is the ``inverse'' of the predecessor relation,  formally,
	\begin{align*}
		\dirsuccsetdef(v, t) \coloneqq \left\{\vapp{w}{t'} \mid \vapp{v}{t}\in\dirpredsetdef(w, t')\right\}.
	\end{align*}
\end{definition}
Clearly, the direct predecessor sets induce a relation over vertex appearances.
We use this to define the following directed graph. We remark that this graph is similar to a so-called static
expansion~\cite{Zsc+20,kempe_connectivity_2002,wu_efficient_2016} that is tailored to a specific source vertex, see \cref{fig:pred-graph} for an example.
\begin{figure}
	\centering
	\begin{tikzpicture}[xscale=1.1,yscale=1.05]
		{\labelledGraph}
		\begin{pgfonlayer}{background}
			\begin{scope}[opacity=.15, transparency group]
				\draw [pathOne] (v1.center) \foreach \i in {2,3,4,8,9} {-- (v\i.center) };  
			\end{scope}
			\begin{scope}[opacity=.15, transparency group]
				\draw [pathTwo,line width=4pt] (v1.center) \foreach \i in {6,4,8,9} {-- (v\i.center) };  
			\end{scope}
			\begin{scope}[opacity=.15, transparency group]
				\draw [pathThree,line width=9pt] (v1.center) \foreach \i in {6,4,8,7,4,8,9} {-- (v\i.center) };  
			\end{scope}
		\end{pgfonlayer}
	\end{tikzpicture}
	\hspace{-0.5cm}
	\begin{tikzpicture}[xscale=0.9,yscale=1.2]
		\foreach[count=\i] \x / \y / \opt in {0/0/{label=below:{\small$(s,0)$}}, 2/0/{label=below:{\small$(d,2)$}}, 5/0/{label=below:{\small$(v,5)$}}, 6/0/{label=below:{\small$(f,6)$}}, 7/0/{label=below:{\small$(e,7)$}}, 8/0/{label=below:{\small$(v,8)$}}, 9/0/{label=below:{\small$(f,9)$}}, 10/0/{label=below:{\small$(z,10)$}}, 1/1/{label=above:{\small$(a,1)$}}, 3/1/{label=above:{\small$(b,3)$}}}
		{
			\node (v\i) at (\x,\y) [draw=black,circle,thick,inner sep=1pt,fill=black,apply style/.expand once=\opt] {};
		}
		\begin{pgfonlayer}{background}
			\foreach \i / \j / \style in {1/2/{}, 2/3/{}, 3/4/{}, 4/5/{}, 5/6/{}, 6/7/{}, 7/8/{}, 3/8/{bend left=40}, 1/9/{}, 9/10/{}, 10/3/{}}
			{
				\draw [->,>=stealth,color=black!75,apply style/.expand once=\style] (v\i) edge (v\j);
			}
			\begin{scope}[opacity=.15, transparency group]
				\draw [pathOne] (v1.center) \foreach \i in {9,10,3} {-- (v\i.center) } edge [bend left=41] (v8.center);  
			\end{scope}
			\begin{scope}[opacity=.15, transparency group]
				\draw [pathTwo,line width=4pt] (v1.center) \foreach \i in {2,3} {-- (v\i.center) } edge [bend left=41] (v8.center);  
			\end{scope}
			\begin{scope}[opacity=.15, transparency group]
				\draw [pathThree,line width=9pt] (v1.center) \foreach \i in {2,...,8} {-- (v\i.center) };  
			\end{scope}
		\end{pgfonlayer}
	\end{tikzpicture}
	\hfil
	\caption{An illustration of the predecessor graph. 
		\emph{Left:} A temporal graph~$\TG=\temptuple$ with a set~$\walkset$ of three highlighted walks. 
		\emph{Right:} The predecessor graph~$\predgraphdef$ of~$s$ with respect to~$\walkset$. 
		The $x$-coordinate of the vertices correspond to the time of the vertex appearance and the walks corresponding to~$\walkset$ are highlighted.
	}
	\label{fig:pred-graph}
\end{figure}
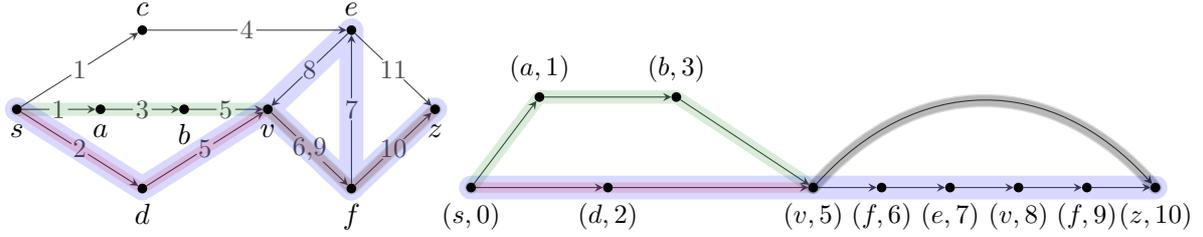

\begin{definition}[Predecessor graph]\label{def:pred-graph} 
	Let~$\TG=\temptuple$ be a temporal graph and $\walkset\subseteq\walks{\TG}$ be a subset of its temporal walks. 
	Fix a source vertex~$s\in V$. 
	Then~$\predgraphdef \coloneqq (U, A)$ is the \emph{predecessor graph (of~$s$, with respect to~$\walkset$)}, where
	\begin{align*}
		U & \coloneqq \{\vappsource\} \cup \{\vapp{w}{t'} \mid \dirpredsetdef(w, t') \neq \emptyset\}, \text{ and} \\
		A & \coloneqq \{(\vappdef, \vapp{w}{t'}) \mid \vappdef\in\dirpredsetdef(w,
		 t')\}.
	\end{align*}
\end{definition}

We next introduce some notation and definitions for temporal walk counting and
temporal betweenness.
\begin{definition}
\label{def:walkcounts}
	Let~$\TG=\temptuple$ be a temporal graph and $\walkset\subseteq\walks{\TG}$ be
	a subset of its temporal walks.
	Let~$s, v, z\in V$ and $t\in[\timespan]$. Then,
	\begin{itemize}
		\item $\sigmaresdef_{sz}$ is the number of~\szwalks{} in~$\walkset$ that start in~$s$ and end in~$z$.
		
		\item $\sigmaresdef_{sz}(v)$ is the number of~\szwalks{} in~$\walkset$ that go
		through the vertex~$v$. Furthermore,~$\sigmaresdef_{sz}(z) =
		\sigmaresdef_{sz}(s) = \sigmaresdef_{sz}$ and~$\sigmaresdef_{ss}(s) = \sigmaresdef_{ss}$;
		
		\item $\sigmaresdef_{sz}(v, t)$ is the number of~\szwalks{} in~$\walkset$ that go through the vertex appearance~$\vapp{v}{t}$. 
		Furthermore,~$\sigmaresdef_{sz}(s, 0) = \sigmaresdef_{sz}(s) =
		\sigmaresdef_{sz}$ and~$\sigmaresdef_{sz}(s, t') = 0$ for all~$t'\in\natinterval{T}$.
	\end{itemize}
\end{definition}

Based on this definition we can define the notions of
dependency of vertices on other vertices, similar to how it was done by
\citet{Bra01} in the static case. We remark that the notions for
the temporal setting introduced in the following are very similar to the ones
used by \citet{BMNR20}. We give them again here for completeness and since we
adapt them for general sets of temporal walks.
\begin{definition}[\Pairdeptext{}, \cumdep{}]\label{def:dependency}
	Let~$\TG$ be a temporal graph and $\walkset\subseteq\walks{\TG}$~be a subset of
	its temporal walks. Then,
	\begin{align*}
 		\pairdepdef(v) & \coloneqq  \begin{cases} 0, & \text{if } \sigmaresdef_{sz} = 0, \\ \frac{\sigmaresdef_{sz}(v)}{\sigmaresdef_{sz}}, & \text{otherwise;}\end{cases} &
 		\deltapointdef(v) & \coloneqq  \sum_{z\in V} \pairdepdef(v)
 	\end{align*}
	are the \emph{\pairdeptext{}} of~$s$ and~$z$ on~$v$ and the \emph{\cumdep{}}
	of~$s$ on~$v$, respectively.
\end{definition}
In other words,~$\pairdepdef(v)$ is the fraction of~\szwalks{} that
go through~$v$. Intuitively, the higher this fraction is, the more
important~$v$ is to the connectivity of~$s$ and~$z$ in the graph.
Furthermore,~$\deltapointdef(v)$ is the \cumdep{} of~$s$ on~$v$ for all
possible destinations. 
These notions can be used to define temporal betweenness
centrality, which intuitively captures how \emph{all} other vertices depend
on~$v$ for their connectivity.
 \begin{definition}[Temporal betweenness centrality]
	\label{def:genbtw}
	 Let~$\TG$ be a temporal graph and~$\walkset\subseteq\walks{\TG}$ be a subset
	 of its temporal walks. 
	Then, for any vertex~$v\in V$, 
	\begin{align*}
		\btwsetdef(v) \coloneqq \sum_{s\neq v\neq z}\pairdepdef(v) && \text{\emph{and}} && \totbtwsetdef(v) \coloneqq \sum_{s, z\in V}\pairdepdef(v)
	\end{align*}
	are the \emph{temporal betweenness centrality}~$\btwsetdef(v)$ of~$v$ and \emph{\totbtwtext{}}~$\totbtwsetdef(v)$ of~$v$ (with respect to $\walkset$).
\end{definition}
If the set of walks~$\walkset$ in question is clear from the context, then we omit the~$\walkset$.
The main reason behind mainly using \totbtwtext{} instead of the standard
temporal betweenness in the following is that it simplifies some of our proofs as it works well
with our definition of \cumdep{}:
\begin{observation}\label{thm:totbtwcumdep}
	For any vertex~$v\in V$, $\totbtwsetdef(v) = \sum_{s\in V} \deltapointdef(v).$
\end{observation}
We have that $\totbtwsetdef(v)$ and $\btwsetdef$ are tightly related: 
\begin{observation}[\kddpaperbetweennesshat{} in \citet{BMNR20}]\label{thm:totbetweennesssame}
	For any vertex~$v\in V$, 
	\begin{align*}
		\btwsetdef(v) = \totbtwsetdef(v) - \sum_{w \in V} \left(\iversonian{\sigmaresdef_{vw} > 0} + \iversonian{\sigmaresdef_{wv} > 0}\right) + \iversonian{\sigmaresdef_{vv} > 0}.
	\end{align*}
\end{observation}

We use a straightforward
generalization of \cref{def:dependency} to vertex appearances.
\begin{definition}[\Temppairdeptext, \cumtdep]\label{def:tempdep}
	Let~$\TG$ be a temporal graph and~$\walkset$ be a subset of its walks. We
	define
		\begin{align*}
		\pairdepdef(v, t) & \coloneqq  \begin{cases} 0, & \text{if } \sigmaresdef_{sz} = 0 \\ \frac{\sigmaresdef_{sz}(v, t)}{\sigmaresdef_{sz}}, & \text{otherwise}\end{cases} & 
		\deltapointdef(v, t) & \coloneqq  \sum_{z\in V} \pairdepdef(v, t)
	\end{align*}
	are the \emph{\temppairdeptext{} of $s$ and $z$ on $\vapp{v}{t}$} and the
	\emph{\cumtdep{} of $s$ on $\vapp{v}{t}$}, respectively.
	Additionally, the special case~$\appdepdef$ denotes the \emph{\appdeptext{} of~$s$ on~$\vapp{v}{t}$}.
\end{definition}

We can observe a simple relation between dependencies on vertices and
dependencies on vertex appearances.
\begin{observation}\label{thm:depappdepsame}
	For any vertex $v\in V$, it holds that
		\begin{align*}
		\pairdepdef(v) & =  \sum_{t\in\natinterval{T}}\pairdepdef(v, t) & \text{\emph{and}} &&
		\deltapointdef(v) & =  \sum_{t\in\natinterval{T}}\deltapointdef(v, t)\text{.}
	\end{align*}
\end{observation}

\section{\Prefcompat}\label{sec:prefcomp}
In this section, our goal is to find
an easy-to-understand-and-use property for optimality concepts for temporal
walks and paths that is sufficient for polynomial-time solvability of
  (1) counting optimal temporal walks and
  (2) computing the temporal betweenness with respect to that optimality concept for temporal walks.
We call this property ``\prefcompat''. 
Intuitively, a class of temporal walks is \prefcomple{} if prefixes of optimal temporal walks are also optimal (``\prefopt'') and prefixes of optimal temporal walks can be exchanged (``\prefext'').
To formally define optimality concepts for temporal walks, we use \critfnstext. 
\begin{definition}[\Critfntext]
	\label{def:critfn}
	Let~$\walkset$ be the set of all temporal walks in a temporal graph~$\TG$.  A function of the form
 		$\critfn : \walkset\rightarrow\RR\cup\{\infty\}$ 
	is a \emph{\critfntext{}}. 
\end{definition}

We remark that for this work we assume that the \critfntext{} can be computed in constant time, which turns out to be a valid assumption for many optimality concepts.
When considering \critfnstext{} that need polynomial time to be evaluated, this polynomial factor would form an extra multiplicative term in our running time results.

Let~$\critfn$ be a \critfntext{} and let $\TG = \temptuple$ be a temporal graph.
Fix a source~$s\in V$.
Then, for every vertex appearance~$(v,t)\in V\times [T]$ we define~$\critfnopt(v, t)$ to be the minimum value of~$\critfn$ assumed over all~\swalks{$s$}{$\vappdef$}. 
That is, we have~$\critfnopt(v, t) = \min_{s\text{-}\vappdef\text{-walk }W} \{\critfn(W)\}$. 
If the minimum is not defined or there is no~\swalk{$s$}{$\vappdef$}, then let~$\critfnopt(v, t) \coloneqq \infty$. 
Similarly, we define the optimal $\critfn$-values for the vertices $v\in V$ as $\critfnopt(v) \coloneqq \min_{t\in\natinterval{\timespan}}\{\critfnopt(v, t)\}$.
We call a \swalk{$s$}{$\vappdef$}~$W$~\emph{\copt{}} if we have~$\critfn(W) = \critfnopt(v, t) < \infty$. 
Similarly, we call a~\swalk{$s$}{$v$}~$W$~\emph{\copt{}} if we have~$\critfn(W) = \critfnopt(v) < \infty$.
Observe that this notion of \copt{} walks is very general and allows to capture essentially all natural walk concepts, see \cref{fig:critfn-illustration} for some examples.

\begin{figure}[t]\centering
	\begin{tikzpicture}[xscale=1.25,yscale=1.2]
		\standardGraph
		\begin{pgfonlayer}{background}
			\begin{scope}[opacity=.15, transparency group]
				\draw [pathOne] (v1.center) \foreach \i in {5,7,9} {-- (v\i.center) };  
			\end{scope}
			\begin{scope}[opacity=.15, transparency group]
				\draw [pathTwo,line width=4pt] (v1.center) \foreach \i in {6,4,8,9} {-- (v\i.center) };  
			\end{scope}
			\begin{scope}[opacity=.15, transparency group]
				\draw [pathThree,line width=9pt] (v1.center) \foreach \i in {6,4,8,7,4,8,9} {-- (v\i.center) };  
			\end{scope}
		\end{pgfonlayer}
	\end{tikzpicture}
	\hfil
	\begin{tikzpicture}[xscale=1.25,yscale=1.2]
		\standardGraph
		\begin{pgfonlayer}{background}
			\begin{scope}[opacity=.15, transparency group]
				\draw [pathTwo] (v1.center) \foreach \i in {2,3,4,8,7,4,8,9} {-- (v\i.center) };  
			\end{scope}
		\end{pgfonlayer}
	\end{tikzpicture}

	\caption{Examples for \copt{} walks for various~$c$.
		\emph{Left:} The top (green) $s$-$z$-path is the shortest $s$-$z$-walk, that is, $c(W)$ is the number of time arcs in the walk~$W$.
		The blue walk and the purple path are the two fastest $s$-$z$-walks in the graph; here $c(W)$ is the difference of the time steps of the first and last time arc in the walk~$W$.
		\emph{Right:} Highlighted is the only \emph{2-restless} $s$-$z$-walk, that is, the difference between the time steps of two consecutive time arcs is at most two~\cite{HMZ20,himmel_efficient_2020}.
		This could be encoded in~$\critfn$ as follows: for a walk~$W = (e_1, \ldots, e_k)$ we have~$c(W) = 1$ if~$t(e_i) + 2 \ge t(e_{i+1})$ for all~$i \in [k-1]$ and~$c(W) = \infty$ otherwise. 
		Notably, in general it is NP-hard to decide whether such a 2-restless $s$-$z$-\emph{path} exists~\cite{HMZ20}, but for walks even the optimization variants (shortest, fastest, $\ldots$) are polynomial-time solvable~\cite{himmel_efficient_2020}. 
	}
	\label{fig:critfn-illustration}
\end{figure}
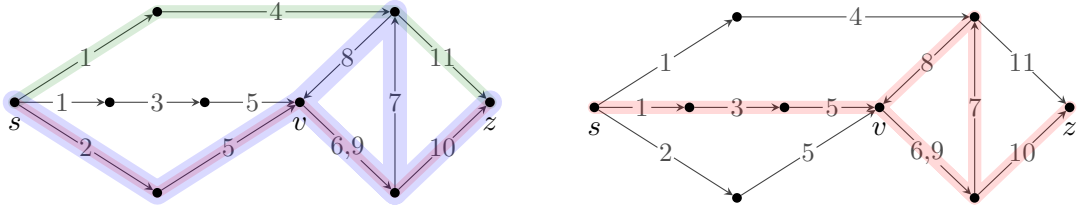

We can now define the set of walks in a temporal graph that is optimal with respect to some \critfntext{}~$\critfn$ in a straightforward way:
\begin{definition}[Induced set of optimal temporal walks]
	Let~$\critfn$ be a \critfntext{} and let $\TG = \temptuple$ be a temporal graph. 
	For~$s, z\in V$,
	let~$\walkset_{sz}$ be the set of all~\szwalks{}
	in~$\TG$. Then
	\begin{align*}
		\walkset^{(\critfn)} \coloneqq \bigcup_{s, z \in V}\{W\in\walkset_{sz}\mid
		\critfn(W) < \infty \wedge \critfn(W) = \critfnopt(z) \}
	\end{align*}
	is the \emph{induced set of optimal temporal walks
	(of~$\critfn$)}.
\end{definition}

From now on, we introduce the two properties for \critfnstext{} that we need to obtain \prefcompat; in \cref{fig:prefix-illustration} we illustrate that (restless) fastest paths do not satisfy one of these.
\begin{figure}[t]\centering
	\begin{tikzpicture}[xscale=1.25,xscale=0.9]
		\foreach[count=\i] \x / \y / \opt in {0/0/{label=below:{$s$}}, 1/0/{}, 2/0/{}, 3/0/{label=below:{$v$}}, 1.5/1/{}, 1.5/-1/{}, 4/1/{}, 4/-1/{},5/0/{label=below:{$z$}}}
		{
			\node (v\i) at (\x,\y) [draw=black,circle,thick,inner sep=1pt,fill=black,apply style/.expand once=\opt] {};
		}
		\begin{pgfonlayer}{background}
			\foreach \i / \j / \txt in {1/2/{1}, 1/5/{\textbf 3}, 1/6/{2}, 2/3/{3}, 3/4/{5}, 4/8/{6,9}, 5/7/{4}, 6/4/{5}, 8/7/{7}, 8/9/{10}, 7/4/{\textbf 5}, 7/9/{\textbf 8}}
			{
				\draw [->,>=stealth,color=black!75] (v\i) edge node[fill=white,edge-label] {\txt}  (v\j);
			}
			\begin{scope}[opacity=.15, transparency group]
				\draw [pathTwo] (v1.center) \foreach \i in {5,7,4} {-- (v\i.center) };  
			\end{scope}
			\begin{scope}[opacity=.15, transparency group]
				\draw [pathThree] (v1.center) \foreach \i in {6,4,8,7,9} {-- (v\i.center) };  
			\end{scope}
		\end{pgfonlayer}
	\end{tikzpicture}
	\hfil
	\begin{tikzpicture}[xscale=1.25,xscale=0.9]
		\foreach[count=\i] \x / \y / \opt in {0/0/{label=below:{$s$}}, 1/0/{}, 2/0/{}, 3/0/{label=below:{$v$}}, 1.5/1/{}, 1.5/-1/{}, 4/1/{}, 4/-1/{},5/0/{label=below:{$z$}}}
		{
			\node (v\i) at (\x,\y) [draw=black,circle,thick,inner sep=1pt,fill=black,apply style/.expand once=\opt] {};
		}
		\begin{pgfonlayer}{background}
			\foreach \i / \j / \txt in {1/2/{1}, 1/5/{\textbf 2}, 1/6/{2}, 2/3/{3}, 3/4/{5}, 4/8/{6,9}, 5/7/{4}, 6/4/{5}, 8/7/{7}, 8/9/{10}, 7/4/{\textbf 5}, 7/9/{\textbf 8}}
			{
				\draw [->,>=stealth,color=black!75] (v\i) edge node[fill=white,edge-label] {\txt}  (v\j);
			}
			\begin{scope}[opacity=.15, transparency group]
				\draw [pathTwo] (v1.center) \foreach \i in {5,7,4} {-- (v\i.center) };  
			\end{scope}
			\begin{scope}[opacity=.15, transparency group]
				\draw [pathThree] (v1.center) \foreach \i in {6,4,8,7,9} {-- (v\i.center) };  
			\end{scope}
		\end{pgfonlayer}
	\end{tikzpicture}

	\caption{Counter examples showing that (restless) fastest paths neither satisfy \prefopt{} nor \prefext{}. 
		Inhere, three time steps (indicated by bold numbers) are updated in our standard temporal graph. 
		\emph{Left:} The blue path (starting at time step 2) is the unique fastest $3$-restless (cf.\ \cref{fig:critfn-illustration}) $s$-$z$-path with travel time~$8-2=6$; it is not \prefoptal{} as the red path is a faster $s$-$(v,5)$-path (travel time~$5-2=3$ vs.\ $5-3=2$).
		\emph{Right:} The blue path (starting at time step 2) is the fastest $s$-$z$-path with travel time~$8-2=6$; it is not \prefextable{} as the red path is also a fastest $s$-$(v,5)$-path but the green prefix cannot be replaced with the red path as the resulting walk would not be a path.
	}
	\label{fig:prefix-illustration}
\end{figure}
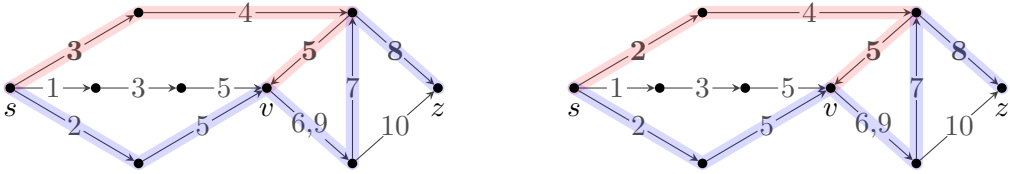
We start with ``\prefopt{}'', which intuitively states that prefixes of optimal temporal walks are also optimal.
\begin{definition}[\Prefopt]
	Let~$\critfn$ be a \critfntext{} and let~$\walkset^{(\critfn)}$ be a set of
	optimal temporal walks in a temporal graph~$\TG$ that is induced by~$\critfn$.
	Then,  $\critfn$~is \emph{\prefoptal{}} if for every
	temporal walk~$W\in\walkset^{(\critfn)}$ and for every vertex appearance~$\vappdef\in
	W$ it holds that 
	$\critfn(\pathpref{W}{\vappdef}) = \critfnopt(v,
	t)$.
\end{definition}

Note that we do not require that the prefixes be optimal temporal walks \emph{to a vertex}, so the temporal walk~$\critfn(\pathpref{W}{\vappdef})$ is not required to be in the induced set of optimal temporal walks~$\walkset^{(\critfn)}$.
If~$\critfn$ is clear from the context and there is no danger of confusion, then we drop the superscript~$(\critfn)$.

The second property we introduce is ``\prefext{}''. It intuitively states that a
prefix of an optimal temporal walk can be exchanged by certain other temporal
walks.
\begin{definition}[\Prefext]
	Let~$\critfn$ be a \critfntext{} and let~$\walkset^{(\critfn)}$ be a set of
	optimal temporal walks in a temporal graph~$\TG$ that is induced by~$\critfn$.
	Then, $\critfn$ is \emph{\prefextable} if for every vertex appearance~$\vappdef\in\appset$ such
	that there exist~$s, z\in V$ for which there is a~\szwalk{}~$W\in
	\walkset^{(\critfn)}$ going through~$\vappdef$ and for
	every~\swalk{$s$}{$\vappdef$}~$W'$ with $\critfn(W') = \critfnopt(v,
	t)$ it holds that
	$W'\pathconcat\pathsuff{W}{\vappdef}\in \walkset^{(\critfn)}$, that is,
	$\critfn(W'\pathconcat\pathsuff{W}{\vappdef}) = \critfnopt(z)$.
\end{definition}
In other words, if there is an optimal \szwalk{}~$W$ going through~$\vappdef$,
then all~\copt{}~\swalks{$s$}{$\vappdef$} can be substituted for the first part of~$W$ to get a~\copt{}~\szwalk{}.

It is convenient to combine the two main properties into one.
\begin{definition}[\Prefcompat]
	Let~$\critfn$ be a \critfntext{}. 
	Then~$\critfn$ is \emph{\prefcomple} if it is both \prefoptal{} and \prefextable{}.
\end{definition}

\subsection{Examples of \prefcomple{} \critfnstext{}} \label{ssec:examples}
We exemplarily show that five well-known optimality concepts for temporal paths
and walks~\cite{wu_efficient_2016,xuan_computing_2003,himmel_efficient_2020} can
be expressed by \prefcomple{} \critfnstext{}. The optimality concepts we consider here are \emph{foremost temporal walks} and
\emph{shortest temporal paths}. 
We further consider
\emph{shortest fastest temporal paths}, which are temporal paths that are
shortest among all fastest temporal paths and \emph{shortest restless temporal
walks}~\cite{himmel_efficient_2020}, which are the shortest temporal walks among
all temporal walks where the difference of the time labels of two consecutive
transitions is bounded by some constant.
Lastly, we consider \emph{strict prefix-foremost temporal paths}~\cite{wu_efficient_2016}. A temporal path is strict prefix-foremost if it is foremost and every prefix is
also foremost and all transitions have increasing time labels.
\begin{proposition}
\label{prop:exps}
The \critfnstext{} describing the following optimality concepts are
\prefcomple{}:
	\begin{itemize}
  \item foremost temporal walk,
  \item shortest temporal path,
  \item shortest fastest temporal path, 
  \item shortest restless temporal walks, and
  \item strict prefix-foremost temporal path.
\end{itemize}
\end{proposition}
\begin{proof}
\begin{description}
	\item{Foremost temporal walks:}
  
  In this case, the \critfntext{} $\critfn$ maps temporal walks to the time label
  of their last transition. It is clear that this \critfntext{} is \prefoptal{}, since
  every \swalk{$s$}{$\vappdef$}, that is, a temporal walk to a vertex
  \emph{appearance}, is trivially \copt{}. The \critfntext{} is also \prefextable{}, since every
  prefix of a foremost temporal walk to some vertex appearance $(v,t)$ can be
  substituted by a different temporal walk to that vertex appearance and the
  overall temporal walk will remain foremost since the arrival time to the last
  vertex is not changed.
  \item{Shortest temporal paths:}
  
  In this case, the \critfntext{} $\critfn$ maps temporal walks to the number of
  transitions in the temporal walk. Similarly to the static case, we can observe
  that every \copt{} temporal walk is in fact a temporal path. We have that
  $\critfn$ is \prefoptal{}, since similarly to the static case we have that
  every prefix of a shortest temporal path is also a shortest temporal path.
  The \critfntext{} $\critfn$ is also \prefextable{}, since substituting a
  prefix by another optimal, that is, equally long, prefix does not change the
  length of the temporal path. Furthermore, we still have a path since
  otherwise, if a vertex is visited multiple times, there is a shorter path that
  waits in the vertex that is visited multiple times, a contradiction to the
  assumption that the original path was shortest.
  \item{Shortest fastest temporal paths:}
  
  In this case, the \critfntext{} $\critfn$ maps temporal walks e.g.\ to the
  duration times $n$ plus the number of transitions in the temporal walk.
  Similarly to the case of shortest temporal paths, we can observe that every
  \copt{} temporal walk is in fact a temporal path. We have that $\critfn$ is
  \prefoptal{}, since a prefix not being \copt{} yields a contradiction to the
  assumption that the original temporal path was \copt{}: we replace the
  presumably non-optimal prefix by an optimal one. Since both prefixes end at
  the same vertex appearance, the optimal one is either faster or it has the
  same duration but is shorter. In both cases substituting the original prefix
  by the optimal one yields a ``better'' (according to~$\critfn$) temporal path
  (if it is not a path but a walk, we can wait at vertices that are visited
  multiple times).
  The \critfntext{} $\critfn$ is also
  \prefextable{}, since substituting a prefix by another optimal prefix does not
  change the duration or the length of the temporal path. Furthermore, we still
  have a path since otherwise, if a vertex is visited multiple times, there is a
  shorter path (with the same duration) that waits in the vertex that is visited
  multiple times, a contradiction to the assumption that the original path was
  shortest fastest.
  \item{Shortest restless temporal walks:}
  
  Let $\Delta$ be the upper bound on the difference of the time labels on two
  consecutive transitions. Similar to shortest temporal paths, the \critfntext{} $\critfn$
  maps temporal walks to the number of transitions in the temporal walk but under the condition that no two consecutive transitions have time labels that differ
  by more than $\Delta$. If the walk contains two consecutive transitions where
  the corresponding time labels differ by more than $\Delta$, then~$\critfn$ maps
  this walk to infinity. 
We have that $\critfn$ is
  \prefoptal{}, since a prefix not being \copt{} yields a contradiction to the
  assumption that the original temporal walk was \copt{}: we replace the
  presumably non-optimal prefix by an optimal one. Since both prefixes end at
  the same vertex appearance, the optimal one is shorter. Now substituting the
  original prefix by the optimal one yields a shorter restless temporal walk. 
  The \critfntext{} $\critfn$ is also \prefextable{}, since substituting a
  prefix by another optimal, that is, equally long prefix does not change the
  length of the temporal walk.
  \item{Strict prefix-foremost temporal paths:}
  
  In this case, we cannot give an explicit description of the \critfntext{}
  $\critfn$ that is independent of the given temporal graph. Nevertheless we
  have that by definition, every prefix of a strict prefix-foremost temporal path
  is strict and prefix-foremost, and hence the corresponding \critfntext{} is
  \prefoptal{}.
  Furthermore, we have that if we exchange a prefix of a strict prefix-foremost
  temporal path with another optimal prefix, then we again have the property that
  every prefix of the obtained path is foremost and we still have a path, since
  otherwise we visit a vertex multiple times at different time steps (since we
  are in the strict case), which is a contradiction to the original temporal
  path being prefix foremost. It follows that the corresponding \critfntext{} is
  also \prefextable{}.
\end{description}
\end{proof}
We remark that e.g.\ ``shortest fastest'' temporal path have to the best of our
knowledge not been considered yet for temporal betweenness computation while
being a very natural optimality criterion. Furthermore, many other natural
optimality concepts can be shown to be \prefcomple{} in a similar way to our
examples.

Since we aim for a very general framework, it is not surprising that our
running times for temporal betweenness computation can be
improved for specific optimality concepts by tailored algorithms.
As we will show in \cref{thm:fordbtw,thm:fordcounting}, we can count \copt{}
temporal walks and compute the temporal betweenness centrality with respect to
optimal walks for \prefcomple{} \critfnstext{}~$\critfn$ in~$O(n^2MT^2)$~time.
However, for example for strict prefix-foremost paths and shortest temporal
paths, 
the corresponding temporal betweenness computation can be done in
$O(nM\log M)$ and $O(n^3T^2)$ time, respectively~\cite{BMNR20}. Also the space
requirements can be improved for specific optimality concepts~\cite{BMNR20}. We
remark that the well-known techniques for static graphs~\cite{bellman1958routing,ford1956network,Bra01} can mostly be
transferred to the temporal setting in straightforward ways.

\subsection{Necessity of \prefopt{} and \prefext{}} \label{ssec:sharpPhard}
We now briefly motivate why we need a \critfntext{}~$\critfn$ to be both
\prefoptal{} and \prefextable{} in order to be able to count \copt{} temporal walks
in polynomial time. 
We do this by giving examples of \critfnstext{} which have
only one of the two properties and where the corresponding problem of counting
\copt{} temporal walks is~\SPC{}. Note that this implies that the corresponding
temporal betweenness computation problem is also~\SPC{}~\cite{BMNR20}. We
remark that this does not imply that \prefopt{} and \prefext{} cannot be replaced
by some weaker requirements.

\begin{proposition}
\label{thm:prefoptred}
	There exists a \prefoptal{} \critfntext{}~$\critfn$ for which counting the
	number of~\copt{} temporal walks is~\SPC{}.
\end{proposition}
\begin{proof}
	Consider the set of foremost temporal paths (strict or non-strict). It is easy
	to see that this set is induced by the following \critfntext{}~$\critfn$:
\begin{align*}
	\critfn(W) \coloneqq \begin{cases}
		t, & \text{$W$ is a (strict) temporal path with arrival time~$t$,} \\
		\infty, & \text{otherwise.}
	\end{cases}
\end{align*}
	Moreover,~$\critfn$ is trivially \prefoptal{} since all temporal paths arriving
	at the same vertex appearance have the same value of the criterion function.
	However, it is known that counting foremost temporal paths is
	\SPC{}~\cite{afrasiabi17,BMNR20}.
\end{proof}

As demonstrated, now we have that if we leave out \prefext{}, then we obtain hardness. Next, we show that if we leave out \prefopt{}, then we also obtain hardness.
\begin{proposition}
\label{thm:prefextred}
	There exists a \prefextable{} \critfntext{}~$\critfn$ for which counting the
	number of~\copt{} temporal walks is~\SPC{}.
\end{proposition}
\begin{proof}
	We reduce~\problemn{Paths} to the problem of counting~\copt{} temporal walks
	for a \prefextable{} \critfntext{}~$\critfn$.
	In~\problemn{Paths} we are given a static graph $G=(V,E)$ and two vertices~$s,z\in V$, and are asked to count the number of different paths from $s$ to
	$z$ in $G$. It is well-known that~\problemn{Paths}
	is~\SPC~\cite{valiant_complexity_1979}.
	
	Consider graph~$G$ and the two vertices~$s, z\in V$. We add a special degree-one
	terminal vertex~$z^{\star}$ to~$V$ and connect it to~$z$. 

	Then,
	\begin{align*}
		\critfn(W) \coloneqq \begin{cases}
			\operatorname{length}(W), & \text{$W$ is a~\swalk{$s$}{$v$} for a~$v\in
			V\setminus\{z^{\star}\}$,} \\
			1, & \text{$W$ is a~\spath{$s$}{$z^{\star}$},} \\
			\infty, & \text{otherwise.}
		\end{cases}
	\end{align*}

	Intuitively, for all~$v\in V\setminus\{z^{\star}\}$ we consider only the
	shortest~\swalks{$s$}{$v$} as optimal (and obviously such shortest walks are
	also paths). This set of shortest walks is also clearly \prefextable{} (and
	\prefoptal{}).
	
	For the special vertex we proceed differently, however. We define
	every~\spath{$s$}{$z^{\star}$} as optimal, whether it is shortest or not. Note
	that clearly all paths to~$z^{\star}$ need to go through~$z$ and, conversely,
	any~\szpath{} can be extended to a~\spath{$s$}{$z^{\star}$}. Hence, the
	number of~\copt{}~\spaths{$s$}{$z^{\star}$} is precisely the number
	of~\szpaths{}. Also note that, in particular, any shortest~\spath{$s$}{$z$} can
	be extended to a~\spath{$s$}{$z^{\star}$}, so our~\critfntext{} will
	remain~\prefextable{}. However, \prefopt{} is now violated since we may
	have a non-shortest~\szpath{} as a subpath of
	a~\copt{}~\spath{$s$}{$z^{\star}$}.
\end{proof}

\section{Counting Walks}\label{sec:walkcounting}
In this section, complementing the hardness shown in \cref{ssec:sharpPhard}, we extend classic algorithms for path and walk counting to our setting.
This will provide a polynomial-time algorithm for counting optimal walks with respect to a \prefcomple{} \critfntext{}.

The general idea is roughly as follows:
first, compute the static predecessor graph~$\predgraphdef(\critfn)$ with respect to~$\critfn$ (see \cref{def:pred-graph}) using a slightly modified version of the classic Bellman-Ford algorithm~\cite{bellman1958routing,ford1956network}.
Second, count the walks in this static graph~$\predgraphdef(\critfn)$ with known approaches; the results correspond to the number of \copt{}~walks in the temporal input graph.

We start with statements explaining the connection between~$\predgraphdef(\critfn)$ and the number of walks in the temporal input graph.
Here, an important corner case is that there might be infinitely many \copt{}~walks.
\begin{definition}[Finiteness]\label{def:finite}
	Let~$\critfn$ be a \critfntext{} for a temporal graph~$\TG$.
	Then, $\critfn$ is \emph{finite} on~$\TG$ if the induced set~$\walkset^{(\critfn)}$ of optimal temporal walks of~$\critfn$ has finite cardinality.
\end{definition}
As stated next, finiteness of the \critfntext~$\critfn$ coincides with the predecessor graph~$\predgraphdef(\critfn)$ containing directed cycles and is, thus, easy to detect.

\begin{lemma}
\label{thm:finiteiffacyclic}
	Let~$\critfn$ be a \prefcomple{} \critfntext{}.
	Then~$\critfn$ is finite if and only if the predecessor graph~$\predgraphdef(\critfn)$ is acyclic. 
\end{lemma}
Before proving \cref{thm:finiteiffacyclic}, we show the following lemma.

\begin{lemma}\label{thm:novisitedtwice}
	Let~$\critfn$ be a \prefcomple{} and finite \critfntext{}.
	Then no \copt{} walk visits the same vertex appearance~$\vappdef$ twice.
\end{lemma}
\begin{proof}
	Assume for the purpose of contradiction that there exists a pair of vertices~$s, z\in V$ such that there exists a~\copt{} walk~$W$ that visits~$\vappdef$ (at least) twice. 
	
	Let now~$W'$ be the prefix of~$W$ that ends at the second time~$\vappdef$ appears in~$W$ and let~$W''$ be the suffix of~$W$ that starts at the first appearance of~$\vappdef$ in~$W$ (see \cref{fig:walkacyclicproof}).
	\tikzset{every loop/.style={min distance=20mm,in=135,out=45,looseness=20}}
	\begin{figure}[t]
		\centering
		\begin{tikzpicture}[scale=2]
			\foreach[count=\i] \x / \y / \opt in {0/0/{}, 2/0/{label=below:{$(v,t)$}}, 4/0/{}}
			{
				\node (v\i) at (\x,\y) [draw=black,circle,thick,inner sep=2pt,fill=white,apply style/.expand once=\opt] {};
			}
			\foreach \i / \j in {1/2, 2/3}
			{
				\draw [->,>=stealth,draw=black] (v\i) -- (v\j);
			}
			\path [->,>=stealth,draw=black] (v2) edge [loop] ();
			\node () at (3,0.75) [text=blue] {$W''$};
			\node () at (1,0.75) [text=red] {$W'$};
			\node () at (2,0.75) {$W$};
			\begin{pgfonlayer}{background}
				\begin{scope}[opacity=.15, transparency group]
					\draw [line join=round, draw=blue,line width=6pt] (v2) edge [loop] () (v2.center) -- (v3.center) ;  
				\end{scope}
				\begin{scope}[opacity=.15, transparency group]
					\draw [line join=round, draw=red,line width=4pt] (v1.center) -- (v2.center) (v2) edge [loop] () ; 
				\end{scope}
			\end{pgfonlayer}
			
		\end{tikzpicture}

		\caption{Illustration of the connection of the walks $W$, $W'$, and~$W''$ in the proof of \cref{thm:novisitedtwice}.}
		\label{fig:walkacyclicproof}
	\end{figure}
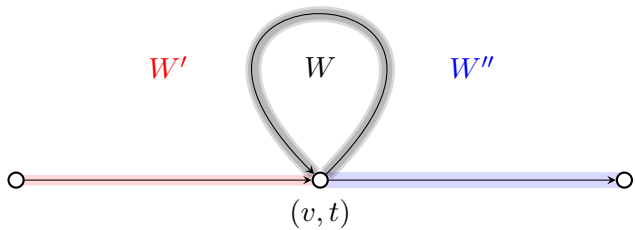
	
	By \prefopt{}, we have~$\critfn(W') = \critfnopt(v, t)$. Hence, by \prefext{}, the walk~$W''' = W' \pathconcat W''$ is also a~\copt{}~\szwalk{}. However, it visits~$\vappdef$ one more time than~$W$ does.
	
	We can now repeat the same procedure with~$W'''$ to get a walk that visits~$\vappdef$ one more time than~$W'''$ does. By carrying the process on inductively, each time adding one more ``loop'' to our walk, we can see that we can visit~$\vappdef$ arbitrarily many times on a~\copt{}~\szwalk{}. This contradicts the finiteness of~$\critfn$.
\end{proof}

Using \Cref{thm:novisitedtwice}, we can prove the aforementioned equivalence:

\begin{proof}[Proof of \cref{thm:finiteiffacyclic}.]
	(``$\Rightarrow$'') Assume for the purpose of contradiction that the predecessor relation is cyclic. Then there exist two vertex appearances~$\vappdef$ and~$\vapp{w}{t}$ such that there exists a walk~$W$ that first visits~$\vappdef$ and then~$\vapp{w}{t}$ and a walk~$U$ that does the opposite. 
	
	Let now~$W'$ be the prefix of~$W$ that ends at the first appearance of~$\vapp{w}{t}$ and~$U'$ the suffix of~$U$ that starts at the first appearance of~$\vapp{w}{t}$. Then, by \prefopt{}, we have~$\critfn(W') = \critfnopt(w, t)$. Hence, by \prefext{}, we have that~$X = W'\pathconcat U'$ is~\copt{}. However, $X$~visits the same vertex appearance~$\vappdef$ twice, contradicting \cref{thm:novisitedtwice}.
	
	(``$\Leftarrow$'') This direction is easy to see: clearly, there are only finitely many walks that do not visit the same vertex appearance more than once. Hence, if the set of~\copt{} walks is infinite, there has to exist a~\copt{} walk that visits some vertex appearance~$\vappdef$ twice. Therefore~$\vappdef$ is its own predecessor (transitively), creating a cycle.
\end{proof}

Assuming that we can efficiently compute~$\predgraphdef(\critfn)$, \cref{thm:finiteiffacyclic} allows us to detect and deal with the cases of infinitely many \copt{} walks.
Moreover, if $\critfn$ is finite, then~$\predgraphdef(\critfn)$ is a DAG and, hence, counting walks is easy.
Hence, we arrive at the following statement.

\begin{lemma}
	\label{thm:genwalkcounting}
	Let~$\critfn$ be a \prefcomple{} \critfntext{} and~$\predgraphdef(\critfn)$ a predecessor graph for a temporal graph~$\TG$ and a vertex~$s$ in~$\TG$.
	Given~$\predgraphdef(\critfn)$, the number of \copt{} temporal walks between from~$s$ to any vertex~$v$ and any vertex appearance~$\vappdef$ in~$\TG$ can be computed in~$O(|\predgraphdef(\critfn)|)$ time.
\end{lemma}
%
{
\begin{proof}[Proof sketch]
	First, run Kosaraju's algorithm \cite{AHU83} on the graph to compute in linear time the strongly connected components (SCCs) in~$\predgraphdef(c)$ while also keeping track of their size. 
	Mark every SCC of size~$>1$ with~$\infty$. 
	Then, run a BFS over the SCCs starting in all nodes marked~$\infty$ and label all the nodes reached during the BFS with~$\infty$.
	Now, for every vertex appearance~$\vappdef$ belonging to an SCC marked with~$\infty$, we set the number of~\swalks{$s$}{$\vappdef$} as~$\infty$ and then remove it from~$\predgraphdef(c)$. 
	The correctness of this step follows from the proof of \cref{thm:finiteiffacyclic}. 
	
	Let~$G'$ be the remaining graph.
	Clearly,~$G'$ is a DAG. 
	We next show that counting paths to a vertex in~$G'$ will exactly correspond to counting \copt{} temporal walks to a vertex appearance corresponding to that vertex:
	
	We shall prove the statement above by induction on the vertices of~$G'$, taken in the topological ordering. 
	First,~$\vappsource$ must clearly be the first vertex in that ordering. Obviously, there is only one~\copt{} path to~$\vappsource$, so the computed value will be correct here.
	
	Now, consider a vertex~$v_i$ corresponding to some appearance~$\vappdef$. 
	By definition, all its direct predecessors~$v_j$ come before~$v_i$ in the topological ordering in the graph. 
	Since~$v_j\in\dirpredsetdef(v_i)$, by \prefext{}, every~\copt{} walk to~$v_j$ can be extended to a~\copt{} walk to~$v_i$. 
	Conversely, by \prefopt{}, we are also not missing any~\copt{} walks to~$\vappdef$. 
	Hence, the computed number of paths to~$v_i$ will also be correct.

	To compute the number of \copt{} temporal walks \emph{to a vertex}~$v\in V$ we can first find~$\critfnopt(v) = \min_{t\in\natinterval{\timespan}}\critfnopt(v, t)$, and then compute~$\sigmaresdef_{sv} = \sum_{t\mid \critfnopt(v, t) = \critfnopt(v)}\sigmaresdef_{s(v, t)}$.
	This path counting in a DAG is clearly doable in time linear in the size of~$\predgraphdef(c)$.
\end{proof}
}
 
To employ \cref{thm:genwalkcounting}, we need to compute~$\predgraphdef(\critfn)$.
To this end, we run a slight variation of the classical Bellman-Ford algorithm~\cite{bellman1958routing,ford1956network}.
This leads to the following lemma.
Recall that we assume here that~$\critfn$ can be evaluated in~$O(1)$ time.

\begin{lemma}
\label{thm:construct-pred-graph}
	Let~$\critfn$ be a \prefcomple{} \critfntext{} for a temporal graph~$\TG$. 
	Let~$s$ be a vertex in~$\TG$.
	Then the predecessor graph~$\predgraphdef(\critfn)$ can be computed in~$\bigO(nMT^2)$ time.
\end{lemma}
%
\begin{proof}
	The algorithm is a straightforward generalization of the Bellman-Ford algorithm~\cite{bellman1958routing,ford1956network}, adapted to our use case, see \cref{alg:bellmanford}. 
	
	The correctness can be easily proven by induction on the (maximal) length of an optimal walk (similar as for the classical Bellman-Ford algorithm).
	Again, as in the proof of \cref{thm:genwalkcounting}, the key observation is that \prefcompat{} ensures that the walks can be extended step by step:
	By \prefext{}, every~\copt{} walk to a vertex appearance~$\vappdef$ can be extended to a~\copt{} walk to a successor, and thus will be found during looping over all arcs (see \cref{algBF:edge-loop,algBF:time-loop} in \cref{alg:bellmanford}). 
	Conversely, by \prefopt{}, we are also not missing any~\copt{} walks to~$\vappdef$. 
	
	As for the running time, it is easy to see that~\cref{alg:bellmanford} runs in~$\bigO(nMT^2)$ time. 
\end{proof}
\begin{algorithm}[t]
	\caption{A generalized Bellman-Ford algorithm for finding the predecessor sets~$\dirpredsetdef$ for a given source~$s\in V$.}
	\label{alg:bellmanford}
	\begin{algorithmic}[1]
		\Input{A temporal graph $\TG = (V,\TE,T)$, a source vertex~$s\in V$ and a \critfntext~$\critfn$.}
		\Output{Predecessor graph~$\predgraphdef(\critfn)$ and the optimal values~$\critfnopt(v, t)$ for all appearances.}
		\State $\dirpredsetdef(v, t) \gets \emptyset; \dist[v, t] \gets \infty$ for all~$v \in V, t \in [T]$ \Comment{Initialization}
		\For{$i = 1$ to $nT$} 
			\For{$(v\trans{t'}w)\in\TE$} \label{algBF:edge-loop}
				\For{$1 \leq t \leq t'$} \label{algBF:time-loop}
					\State $\Call{Relax}{\vappdef, \vapp{w}{t'}}$
				\EndFor
			\EndFor
		\EndFor
		\State \Return $\dirpredsetdef$, $\dist$
		\end{algorithmic}
\end{algorithm}
\begin{algorithm}[t]
	\caption{Function relaxing the transition from~$\vapp{v}{t}$ to~$\vapp{w}{t'}$. Note that by~$\pathenc{v, t}$ we mean an arbitrary~\swalk{$s$}{$\vappdef$} that is implicitly represented by~$\dirpredsetdef(v, t)$. (By \prefcompat{} all such walks are effectively interchangeable.)}
	\begin{algorithmic}[1]
		\Function{Relax}{$\vapp{v}{t}$, $\vapp{w}{t'}$}
		\If {$\dist[v, t] = \infty$}
		\State \Return
		\EndIf
		\If {$\critfn(\pathenc{v, t}\pathconcat(v\trans{t'}w)) < \dist[w, t']$}
		\State $\dirpredsetdef(w, t') \gets \emptyset$
		\State $\dist[w, t'] \gets \critfn(\pathenc{v, t}\pathconcat(v\trans{t'}w))$
		\EndIf
		\If {$\critfn(\pathenc{v, t}\pathconcat(v\trans{t'}w)) = \dist[w, t']$}
		\State $\dirpredsetdef(w, t') \gets \dirpredsetdef(w, t') \cup \{\vappdef\}$
		\EndIf
		\EndFunction
	\end{algorithmic}
\end{algorithm}
%
Applying \cref{thm:genwalkcounting,thm:construct-pred-graph} starting from each vertex yields the following.

\begin{theorem}[Walk counting]\label{thm:fordcounting}
	Let~$\critfn$ be a \prefcomple{} \critfntext{} for a temporal graph~$\TG = \temptuple$. 
	Then the number of~\copt{} temporal walks from each vertex~$s \in V$ to any vertex appearance~$\vappdef$ can be computed in~$\bigO(n^2MT^2)$ time.
\end{theorem}

\section{Computing Temporal Betweenness}\label{sec:tempbet}
In this section, we discuss how to compute temporal betweenness centrality
efficiently for optimal temporal walks defined by a \prefcomple{} \critfntext{}.
In order to do this, we adapt the machinery of showing a recursive relation of
the temporal dependencies~\cite{BMNR20} to our generalized setting. 
Together with the fact that we
can compute the walk-counts in polynomial time (\cref{thm:fordcounting}), we can
use a Brandes-like~\cite{Bra01} approach to compute the temporal betweenness
values in polynomial time.

We prove a general dependency accumulation formula, alongside some of its implications. 
Note that we essentially adapt the results of \citet{BMNR20} to our generalized setting, the proofs are quite similar. 
In the following, for convenience we use~$\walkset$ (instead of~$\walkset^{(\critfn)}$) to denote a set of~\copt{} walks for some arbitrary \critfntext{}~$\critfn$.
As betweenness centrality is not well-defined for infinite \critfntext{}s, we assume here that~$\critfn$ is finite (in the sense of \cref{def:finite}). 
Note that we can detect whether~$\critfn$ is finite in polynomial time (see \cref{sec:walkcounting}).
First, we define ``\edgedeptext{}'' which we need for the general dependency
accumulation lemma and we show how to compute it.
\begin{definition}[\Edgedeptext]\label{def:edgedep}
	$\pairdepdef(v, t, (v, w, t'))$ denotes the fraction of~\szwalks{} in~$\walkset$ that go through the appearance~$\vapp{v}{t}$ and use the temporal arc~$(v, w, t')$.
\end{definition}
\begin{lemma}
\label{thm:dependencyAccumulationFractionLemma}
Let~$\critfn$ be a finite \prefcomple{} \critfntext{} and let $\TG =
	 \temptuple$ be a temporal graph.
Let~$\walkset$ be the set of optimal temporal walks in
 $\TG$ induced by~$\critfn$. Fix a source~$s\in V$. If $\pairdepdef(v, t, (v, w, t'))$ is positive, then 
	\begin{align*}
		\pairdepdef(v, t, (v, w, t')) = \frac{\sigmaresdef_{s\vappdef}}{\sigmaresdef_{s\vapp{w}{t'}}}\cdot\frac{\sigmaresdef_{sz}(w, t')}{\sigmaresdef_{sz}}.
	\end{align*}
\end{lemma}
\begin{proof}
	Let~$W$ be an arbitrary~\copt{}~\szwalk{} that goes through~$\vappdef$ and makes the direct transition~$v\trans{t'}w$. By~\prefopt{}, we have that~$\critfn(\pathpref{W}{\vappdef}) = \critfnopt(v, t)$ and~$\critfn(\pathpref{W}{\vapp{w}{t'}}) = \critfnopt(w, t')$. Now, by \prefext{}, we can substitute an arbitrary~\copt{}~\swalk{$s$}{$\vappdef$} for $\pathpref{W}{\vappdef}$ and still get a valid~\copt{}~\swalk{$s$}{$\vapp{w}{t'}$}. Similarly, we can then substitute an arbitrary~\copt{}~\swalk{$s$}{$\vapp{w}{t'}$} for $\pathpref{W}{\vapp{w}{t'}}$ and get a valid~\copt{}~\szwalk{}. 
	
	\begin{figure}[t]
	\centering
	
	\begin{tikzpicture}[scale=0.85,yscale=1,xscale = 1,
		vertex/.style={circle, draw, minimum size=4em}, transform shape]
		\path
		node[vertex](s) {$s$}
		node[right = of s](sd) {} 
		node[vertex, right = of sd](v) {$\vapp{v}{t}$}
		node[vertex, right = of v](w) {$\vapp{w}{t'}$}
		node[right = of w](wd) {} 
		node[vertex, right = of wd](z) {$z$};
		
		\draw [->, decorate, decoration={snake}, thick] (s)--(v) node [midway, above] {\Large$\sigmaresdef_{s\vappdef}$};
		\draw (s) edge [->, decorate, decoration={snake, post length=0.5mm}, thick, out=+60,in=+120] node [midway, above=0.25em] {\Large$\sigmaresdef_{s\vapp{w}{t'}}$} (w);
		\draw (v) edge [thick] (w);
		\draw [->, decorate, decoration={snake}, thick] (w)--(z) node [midway, above] {\Large$\frac{\sigmares{\pathset}_{sz}(w, t')}{\sigmaresdef_{s\vapp{w}{t'}}}$};
	\end{tikzpicture}
	\caption{\cref{thm:dependencyAccumulationFractionLemma}: combining arbitrary prefix to~$\vapp{v}{t}$ with an arbitrary suffix from~$\vapp{w}{t'}$.}
	\label{fig:edgedep}
\end{figure}
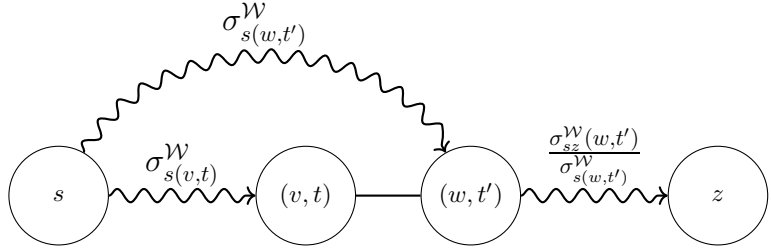
	
	Hence, we can combine an arbitrary prefix which ends in~$\vapp{v}{t}$ with an arbitrary suffix which starts at~$\vapp{w}{t'}$ to get a valid~\szpath{} in~$\pathset$ (see \cref{fig:edgedep}). With all of that in mind, we can now make the following argument:
	
	Of the~$\sigmaresdef_{s\vapp{w}{t'}}$ many walk prefixes which end in~$\vapp{w}{t'}$, exactly~$\sigmaresdef_{s\vappdef}$ go through~$\vapp{v}{t}$ and use the transition~$v\trans{t'}w$. Now, there are also~$\frac{\sigmaresdef_{sz}(w, t')}{\sigmaresdef_{s\vapp{w}{t'}}}$ many unique walk suffixes starting from the appearance~$\vapp{w}{t'}$ and going to the vertex~$z$. Thus, there are~$\sigmaresdef_{s\vappdef}\cdot\frac{\sigmaresdef_{sz}(w, t')}{\sigmaresdef_{s\vapp{w}{t'}}}$ many~\szwalks{} that go through~$\vapp{v}{t}$ and use the transition~$v\trans{t'}w$. Finally, we divide by the total number of~\szwalks{} to get the expression in the statement of the lemma.
\end{proof}
The lemma allows to prove the following general dependency accumulation formula.

\begin{lemma}
[General dependency accumulation]\label{thm:gendepacc} 
Let~$\critfn$ be a finite \prefcomple{} \critfntext{} and let $\TG =
	 \temptuple$ be a temporal graph.
Let~$\walkset$ be the set of optimal temporal walks in
 $\TG$ induced by~$\critfn$. Fix a source~$s\in V$. Then, 
	\begin{align*}
		\deltapointdef(v, t) = \appdepdef + \sum_{\vapp{w}{t'}\in \dirsuccsetdef(v, t)}\frac{\sigmaresdef_{s\vappdef}}{\sigmaresdef_{s\vapp{w}{t'}}}\cdot\deltapointdef(w, t')
	\end{align*}
\end{lemma}
\begin{proof}
	We start by expanding the sum in the formula defining~$\deltapointdef(v, t)$, 
	considering its summands a bit more precisely.
	\begin{align}
		\deltapointdef(v, t) & = \sum_{z\in V}\pairdepdef(v, t) \notag \\ 
		& = \appdepdef + \sum_{z\in V}\sum_{\vapp{w}{t'}\in \dirsuccsetdef(v, t)} \pairdepdef(v, t, (v, w, t')) \label{eq:gendepexpansion} 
	\end{align}
	Where~$\pairdepdef(v, t, (v, w, t'))$ is as defined in \cref{def:edgedep}. Since~$\vapp{v}{t}$ is not a direct predecessor on any~\swalk{$\sourcevertex$}{$\vappdef$}, we need to pull it out of the sum. Conversely, for any~$z\in V$, the vertex appearance~$\vappdef$ may be contained at most once in any \szwalk{} (by \cref{thm:novisitedtwice}). Hence the inner sum of \cref{eq:gendepexpansion} precisely captures $\pairdepdef(v, t)$.
	
	We can now use \cref{thm:dependencyAccumulationFractionLemma} to simplify the summation formula:
		\begin{align*}
		&\sum_{z\in V}  \sum_{\vapp{w}{t'}\in \dirsuccsetdef(v, t)} \pairdepdef(v, t, (v, w, t')) \\
		\overset{\cref{thm:dependencyAccumulationFractionLemma}}{=} & \sum_{z\in V}
		\sum_{\vapp{w}{t'}\in \dirsuccsetdef(v, t)}\frac{\sigmaresdef_{s\vappdef}}{\sigmaresdef_{s\vapp{w}{t'}}}\cdot\frac{\sigmaresdef_{sz}(w, t')}{\sigmares{\walkset}_{sz}} \\
		= & \sum_{\vapp{w}{t'}\in \dirsuccsetdef(v, t)}  \sum_{z\in V}\frac{\sigmaresdef_{s\vappdef}}{\sigmaresdef_{s\vapp{w}{t'}}}\cdot\frac{\sigmaresdef_{sz}(w, t')}{\sigmares{\walkset}_{sz}} \\
		= & \sum_{\vapp{w}{t'}\in \dirsuccsetdef(v, t)}  \frac{\sigmaresdef_{s\vappdef}}{\sigmaresdef_{s\vapp{w}{t'}}} \sum_{z\in V}\frac{\sigmaresdef_{sz}(w, t')}{\sigmares{\walkset}_{sz}} \\
		\overset{\cref{def:tempdep}}{=} & \sum_{\vapp{w}{t'}\in \dirsuccsetdef(v, t)}  \frac{\sigmaresdef_{s\vappdef}}{\sigmaresdef_{s\vapp{w}{t'}}}\cdot\deltapointdef(w, t'),
	\end{align*}
	from which the result immediately follows.
\end{proof}


We now combine \cref{thm:gendepacc} with the results from \cref{sec:walkcounting} to show how betweenness centrality can be computed for all finite \prefcomple{} \critfnstext{}.

\begin{lemma}
\label{thm:genbtwcomp}
	Let~$\critfn$ be a finite \prefcomple{} \critfntext{} for a temporal graph~$\TG = \temptuple$.
	Given~$\predgraphdef(\critfn)$ for each~$s\in V$, the temporal betweenness centrality of all vertices in~$\TG$ can be computed in~$\bigO(\sum_{s\in V}|\predgraphdef(\critfn)| + nM)$ time.
\end{lemma}
\begin{proof}
	
	Consider \cref{alg:metabrandes}. 
	\countwalksfn{} counts optimal temporal walks in a manner described by \cref{thm:genwalkcounting}. 
	We first prove the correctness of the algorithm before analyzing its running time.
	Denote with~$\walkset$ the set of optimal temporal walks in~$\TG$ induced by~$\critfn$.

\begin{algorithm}[t]
	\caption{
		General betweenness algorithm, see \cref{thm:genbtwcomp}. 
		It uses an auxiliary function \countwalksfn{} that computes the walk counts in a manner described by \cref{thm:genwalkcounting}.
		Note that the predecessor sets~$\dirpredsetdef$ are encoded in the predecessor graphs~$\predgraphdef(\critfn)$.
	}\label{alg:metabrandes}
	\begin{algorithmic}[1]
		\Input{A temporal graph $\TG = (V,\TE,T)$, the predecessor graphs~$\predgraphdef(\critfn)$ for all~$s\in V$.}
		\Output{Betweenness $\btwsetdef(v)$ of all vertices $v \in V(\TG)$.}
		\State $\sigmaresdef_{sv}, \sigmaresdef_{s(v, t)}, \appdepdef \gets \countwalksfn(\predgraphdef(\critfn))$ \Comment{\cref{thm:genwalkcounting}}
		\For{$v \in V$}
		\State $\btwsetdef[v] \gets 1$ \Comment{Initialize to $1$ per \cref{thm:totbetweennesssame}} \label{algln:btwinit}
		\EndFor
		\For{$s \in V$} \label{algln:mainloop}
		\For{$(u, v, t) \in \TE$} \label{algln:depinit}
		\State $\deltapointdef[v, t] \gets 0$ \Comment{Reset the array}
		\EndFor
		\For{$\vapp{w}{t'}\in\relverarr$ in topological order determined by~$\dirpredsetdef$}\label{algln:accloop}
		\State $\deltapointdef(w, t') \gets \deltapointdef(w, t') + \pairdep{\pathset}{s}{w}(w, t')$ \Comment{\Appdeptext{} on~$\vapp{w}{t'}$}
		\For{$(v,t) \in \dirpredsetdef(w,t')$}
		\State $\deltapointdef[v,t] \gets \deltapointdef[v,t] 
		+ \frac{\sigmaresdef_{s(v, t)}}{\sigmaresdef_{s(w, t')}}\cdot \deltapointdef[w,t']$ \Comment{Sum of \cref{thm:gendepacc}}
		\State $\btwsetdef[v] \gets \btwsetdef[v] + \frac{\sigmaresdef_{s(v, t)}}{\sigmaresdef_{s(w, t')}}\cdot \deltapointdef[w,t']$ \label{algln:btwacc}
		\EndFor
		\EndFor
		\State $\btwsetdef[s] \gets \btwsetdef[s] - \lvert\{v\mid\exists t\in\natinterval{T} \colon \pairdep{\pathset}{s}{v}(v, t) > 0\}\rvert$ \Comment{Connectivity correction} \label{algln:conncorrection}
		\EndFor
		\State \Return $\btwsetdef$
	\end{algorithmic}
\end{algorithm}

	\smallskip
	
	\noindent\emph{Correctness:} The general idea of the algorithm is to use
	\cref{thm:gendepacc} to implicitly compute the \totbtwtext{} and use
	\cref{thm:totbetweennesssame} to recover the $\btwsetdef$ values. As the
	equation in \cref{thm:totbetweennesssame} has the constant
	summand~$+1$, in \cref{algln:btwinit} we initialize the temporal betweenness
	values to~$1$.
	
	The next step is to compute the \cumdeps{}~$\deltapointdef(v, t)$ for each
	source vertex~$s$ and appropriately update the temporal betweenness values. We
	do this in the loop starting on \cref{algln:mainloop}. We first initialize the array
	holding the \cumdeps{} on \cref{algln:depinit}.
	
	Finally, in the loop starting on \cref{algln:accloop} we compute the
	\cumdeps{} using the recursive formula of \cref{thm:gendepacc}. We proceed in
	reverse topological order, that is, we start with vertices that have no
	successors (and hence for which the equation in \cref{thm:gendepacc} is
	trivial to evaluate) as our base case and then proceed backwards. This is
	possible as finite \prefcomple{} \critfnstext{} always lead to acyclic
	predecessor graphs (see \cref{thm:finiteiffacyclic}).
	
	Finally, on \cref{algln:conncorrection} we apply a ``connectivity
	correction.'' This term corresponds to the $\sum_{w \in V}
	\iversonian{\sigmaresdef_{vw} > 0}$ part of the equation in
	\cref{thm:totbetweennesssame}. Notice that we do not have to handle the term
	with~$\iversonian{\sigmaresdef_{wv} > 0}$ as on \cref{algln:btwacc} we only
	ever add terms which correspond the sum of the equation of
	\cref{thm:gendepacc}, and never the~$\appdepdef$ terms. (We do, however, add
	those terms to the~$\deltapointdef(v, t)$ values which then propagate into
	the~$\btwsetdef$ array, which is why we need the correction on
	\cref{algln:conncorrection}.)
	
	\smallskip
	
	\noindent\emph{Running time:} 
	By \cref{thm:genwalkcounting}, $\countwalksfn{}$ can be computed in $O(|\predgraphdef(\critfn)|)$ time for each~$s \in V$.
	We can easily see that for the loop on \cref{algln:depinit} we need~$\bigO(M)$ time overall.
	
	Before the execution of the for-loop on \cref{algln:accloop}, we may first
	need to compute the topological order on the predecessor graph. This can be
	done in time linear in the size of the graph. The sum of the sizes of the
	graphs for all sources is~$\sum_{s\in V}|\predgraphdef(\critfn)|$. We also note that
	the for-loop on \cref{algln:accloop} goes over each arc of the predecessor
	graph exactly once, so we get an~$\bigO(\sum_{s\in V}|\predgraphdef(\critfn)|)$ bound for the overall execution time.
	
	Finally, with some bookkeeping done in the loop of \cref{algln:accloop},  the
	connectivity correction of \cref{algln:conncorrection} can easily be made to
	run in constant time.
	
	Altogether, we get $\bigO(\sum_{s\in V}|\predgraphdef(\critfn)| + nM)$ for the total running time.
\end{proof}

Using the running time bound from \cref{thm:construct-pred-graph} together with \cref{thm:genbtwcomp}, we immediately get our main result of this work.
\begin{theorem}[General betweenness computation]\label{thm:fordbtw} 
	Let~$\critfn$ be a finite \prefcomple{} \critfntext{}. 
	Then the betweenness centrality of all vertices can be computed in~$\bigO(n^2MT^2)$ time.
\end{theorem}
Combining \cref{prop:exps,thm:fordbtw} yields the following result. 
Note that for all of those optimality concepts \cref{alg:bellmanford} can easily be implemented so that~$c$ can evaluated in amortized constant time, hence our results apply directly without any additional multiplicative factors in the running time.
\begin{corollary}
	The betweenness centrality of all vertices in a temporal graph can be computed
	in~$\bigO(n^2MT^2)$ time with respect to:
		\begin{itemize}
  \item foremost temporal walk,
  \item shortest temporal path,
  \item shortest fastest temporal path, 
  \item shortest restless temporal walks, and
  \item strict prefix-foremost temporal path.
\end{itemize}
\end{corollary}
We remark that, while foremost temporal walk, shortest temporal path, and strict prefix-foremost temporal path were known from previous work \cite{BMNR20,afrasiabi17,kim_temporal_2012}, this is a new classification for shortest fastest temporal paths and shortest restless temporal walks.

\section{Conclusion}
The very nature of this work is conceptual. It goes without saying
that to achieve improved efficiency, exploiting specific properties 
of the various temporal path and walk concepts may clearly 
allow for further improved polynomial running times. 
As to future research, we wonder whether our concept of prefix-compatibility 
may finally lead to a full characterization of polynomial-time 
computable temporal betweenness centrality values.
As to the computationally hard cases (but not only them),
for high efficiency in practice, one might also explore the possibilities of
efficient data reductions or approximation algorithms. 
This proved useful in the statci graphs case, with respect to 
data reduction~\cite{BGPL12,BDKNN20,PEZDB14,SKSC17} as well as with respect to approximation~\cite{BKMM07,GSS08,RK16}

\bibliographystyle{abbrvnat}
\bibliography{bibliography}

%

\end{document}